\documentclass[3p]{elsarticle}

\usepackage{amssymb}

\usepackage[overload]{empheq}
\usepackage{latexsym,dsfont,textcomp,amsmath}
\usepackage[english]{babel}
\usepackage[latin1]{inputenc}
\usepackage{enumerate,indentfirst}
\usepackage[colorlinks,citecolor=blue,urlcolor=blue]{hyperref}
\usepackage{amssymb}
\usepackage{amsthm}
\usepackage{hyperref}

\newtheorem{Theorem}{Theorem}[part]

\newtheorem{Proposition}{Proposition}[part]

\newtheorem{Remark}{Remark}[part]

\newcommand{\nc}{\newcommand}
\nc{\esssup}{\mathop{\mathrm{ess\,sup}}}
\nc{\essinf}{\mathop{\mathrm{ess\,inf}}}
\nc{\argmax}{\mathop{\mathrm{arg\,max}}}

\def \R{\mathbb{R}}

\def \E{\mathbb{E}}
\def \F{\mathbb{F}}

\def \1{\mathds{1}}

\def \Fc{{\cal F}}
\def \Gc{{\cal G}}

\def \Nc{{\cal N}}
\def \Sc{{\cal S}}
\def \Tc{{\cal T}}

\def \l({{\left (}}
\def \r){{\right )}}
\def \l[{{\left [}}
\def \r]({{\right ]}}

\begin{document}

\begin{frontmatter}



\title{Markov switching quadratic term structure models\\\vspace{1cm}\large{12 May 2013}}


\author{St\'ephane GOUTTE}
\address{Laboratoire de Probabilit\'es et Mod\`eles Al\'eatoires,
CNRS, UMR 7599, Universit{\'e} Paris Diderot 7.}
\ead{goutte@math.univ-paris-diderot.fr}

\begin{abstract}
In this paper, we consider a discrete time economy where we assume that the short term interest rate follows a quadratic term structure of a regime switching asset process. The possible non-linear structure and the fact that the interest rate can have different economic or financial trends justify the interest of  Regime Switching Quadratic Term Structure Model (RS-QTSM). Indeed, this regime switching process depends on the values of a Markov chain with a time dependent transition probability matrix which can well captures the different states (regimes) of the economy. We prove that under this modelling that the conditional zero coupon bond price admits also a quadratic term structure. Moreover, the stochastic coefficients which appear in this decomposition satisfy an explicit system of coupled stochastic backward recursions.
\end{abstract}

\begin{keyword}
Quadratic term structure model; Regime switching; Zero coupon bond; Markov chain.

\vspace{0.2cm}
\MSC[2010] 60J10 91B25 91G30\\
\JEL G10, G11
\end{keyword}
\end{frontmatter}
\vspace{-0.6cm}
\section*{Introduction}
Modeling the term structure of interest rates has long been an important topic in economics and finance. Most of the papers about modelling of the interest rate term structure are relative to the family of the Affine Term Structure Models (ATSM). This modelling considers a linear relation between the log price of a zero coupon bond and its states factors. Those models have been first studied by Vasicek (1977) in \cite{VAS77} and Cox, Ingersoll and Ross (1985) in \cite{CIR85}. Then developed by Duffie and Kan (1996) in \cite{DK96} and Dai and Singleton (2000) in \cite{DS00}. A first extension of this class of model was to use regime switching model. Thus, Elliott et al. (2011) in \cite{ESB11} considered a discrete-time, Markov, regime-switching, affine term structure model for valuing bonds and other interest-rate securities. Recently, Goutte and Ngoupeyou (2013) in \cite{GN2} obtained explicit formulas to price defaultable bond under this class of regime switching models. The proposed model incorporates the impact of structural changes in economic conditions on interest rate dynamics and so can capture different economics (financials) levels or trends of the economy. A second extension was to not only consider a linear model. Thus to model the term structure of interest rates with Quadratic Term Structure Models (QTSM). This family, first introduce by Beaglehole and Tammey (1991) in \cite{BT91} are applied to price contingent claims (Lieppold and Wu (2002) in \cite{LW02}) and to the credit risk pricing (Chen, Filipovic and Poor (2004) in \cite{CFP04}). Hence, in this paper we propose to use both of the previous extension and so a regime switching discrete-time version of quadratic term structure models (RS-QTSM).

An important application of term structure models is the valuation of interest rate instruments, such as zero coupon bonds. We will prove that under the regime switching quadratic term structure modeling that the conditional zero coupon bond price of a regime switching asset admits also a quadratic decomposition. Moreover, we will find that the stochastic coefficients which appear in this decomposition satisfy an explicit system of coupled stochastic backward recursions.

This paper is then organized as follows. In section 1, the model is presented and defined. Then in Section 2, the conditional zero coupon bond price is evaluated and we give the corresponding system of coupled stochastic backward recursions. 

\section{The model}
\quad We consider a discrete time economy with finite time horizon and time index set $\Tc:=\{k\vert k=0,1,2,\dots,T\}$, where $T$ is a positive integer such that $T<\infty$. Let $(\omega, \Fc,P)$ be a filtered probability space where $P$ is a risk neutral probability.
\subsection{Markov chain}
Following Elliott et al. in \cite{Ell94}, let $(X_k)_{k\in \Tc}$ be a discrete time Markov chain on finite state space $\Sc:=\{e_1,e_2,...,e_N\}$, where $e_i$ has unity in the $i^{th}$ position and zero elsewhere. Thus $\Sc$ is the set of canonical unit column vectors of $\R^N$. In an economic point of view, $X_k$ can be viewed as an observable exogenous quantity which can reflect the evolution of the state of the economy. We assume that the time dependent transition probability matrix $Q_k:=(q_{ijk})_{i,j=1,...,N}$ of $X$ under $P$ is defined by
$$
q_{ijk}=P\left(X_{k+1}=j\vert X_k=i\right).
$$
It also satisfies $q_{ijk}\geq 0$, for all $i\neq j \in \Sc$ and $\sum_{j=1}^Nq_{ijk}=1$ for all $i \in \Sc$. Let $\F^X=(\Fc^X_k)_{k\in \Tc}:=\sigma(X_k, k\in \Tc)$ which is the $P$ augmented filtration generated by the history of the Markov chain $X$ and $\Fc^X_k$ is the $P$-augmented $\sigma$-field generated by the history of $X$ up to and including time $k$. Moreover, following again Elliott et al. in \cite{Ell94} 
 , we have that the semi-martingale decomposition for the Markov chain $X$ is given by
$$
X_{k+1}=Q_kX_k+M^X_{k+1},\quad k \in \{0,1,2,\dots,T-1\},
$$
where $(M^X_k)$ is an $\R^N$-valued martingale increment process (i.e. $\E\left[M^X_{k+1}\vert \Fc^X_k\right]=0$).
\subsection{Asset}
Let $(S_k)_{k\in \Tc}$ denotes the state asset process and we denote by $\F^S=(\Fc^S_k)_{k\in \Tc}$ the $P$-augmented filtration generated by the process $S$. Finally, we denote by  $\Gc_k:=\F^S_k\vee \F^X_k$ the global enlarged filtration for all $k\in \Tc$.
Let $\langle.,.\rangle$ denote the inner product in $\R^N$. Then, for every $k\in \{1,2,\dots,T\}$, we define the following regime dependent parameters $\kappa_k:=\kappa(k,X_k)=\langle \kappa,X_k\rangle$ ,$\mu_k:=\mu(k,X_k)=\langle \mu,X_k\rangle$ and $\sigma_k:=\sigma(k,X_k)=\langle \sigma,X_k\rangle$ where $\kappa:=(\kappa_1,\kappa_2,\dots,\kappa_N),\mu:=(\mu_1,\mu_2,\dots,\mu_N)$ and $\sigma:=(\sigma_1,\sigma_2,\dots,\sigma_N)$ are $1\times N$ real-valued vectors. Moreover we assume that $\sigma_i>0$, for all $i\in \{1,2,\dots,N\}$. Finally, $\varepsilon:=(\varepsilon_k)_{k\in \{1,2,\dots,T\}}$ are a sequence of independent and identically distributed random variables with law $\Nc(0,1)$. We assume that $\varepsilon$ and the Markov chain $X$ are independent.
We then have that under the risk neutral probability measure $P$ the dynamic of the asset $S$ is governed by the following discrete time, Markov switching model
\begin{equation}\label{EqSDis}
S_{k+1}=\kappa_k+\mu_kS_k+\sigma_k\varepsilon_{k+1},\quad k=\{0,1,\dots,T-1\}.
\end{equation}

\subsection{Short term interest rate}
Let $(r_k)_{k\in \Tc}$ denote the process of short term interest rate. We assume that the dynamic of $r_k$ is regime dependent and is following a quadratic term structure of the asset process $S_k$ which is given by
\begin{equation}\label{EqR}
r_k:=r(k,X_k)=a_{0,k}+a_{1,k}S_k+a_{2,k}S^2_k, \quad k\in \Tc.
\end{equation}
with $r_k:=r(k,X_k)=\langle r,X_k\rangle$, $r:=\left(r_1,r_2,\dots,r_N\right)$, $a_{0,k}:=a_0(k,X_k)=\langle a_0,X_k\rangle$, $a_{1,k}:=a_1(k,X_k)=\langle a_1,X_k\rangle$ and $a_{2,k}:=a_2(k,X_k)=\langle a_2,X_k\rangle$ where $a_0,a_1$ and $a_2$ are real vectors of size $1\times N$.

\subsection{Zero-coupon Bond price}

Let $P(k,T)$ be the price at time $k\in \Tc$ of a zero-coupon bond with maturity $T$. Since we are under the risk neutral probability, we have that 
\begin{equation}\label{EqP1}
P(k,T)=\E\left[\exp\left(-\sum_{t=k}^{T-1}r_t\right)\vert \Gc_k\right], \quad k\in \Tc,
\end{equation}
with $P(T,T)=1$ and $P(T-1,T)=\exp\left(-r_{T-1}\right)$.


\section{Regime switching quadratic structure formulas}

\subsection{Full history case}
Assume firstly that we know the full history of the Markov chain $X$. Thus, we denote  by $\tilde{\Gc}_k:=\Fc_T^X\vee \Fc^S_k$, $k\in \Tc$ this enlarged information set. Then we denote by $\tilde{P}(k,T)$ the conditional zero coupon bond price at time $k$ with maturity $T$ given the enlarged filtration $\tilde{\Gc}_k$. We obtain that
\begin{equation}\label{EqP2}
\tilde{P}(k,T)=\E\left[\exp\left(-\sum_{t=k}^{T-1}r_t\right)\vert \tilde{\Gc}_k\right], \quad k\in \Tc,
\end{equation}
with $\tilde{P}(T,T)=1$ and $\tilde{P}(T-1,T)=\exp\left(-r_{T-1}\right)$.

\begin{Theorem}\label{PropoFullHis}
The conditional bon price $\tilde{P}(k,T)$ has an exponential quadratic term structure given for all $k\in \Tc$ by
\begin{equation}\label{EqPropoFullHis}
\tilde{P}(k,T)=\exp\left\{c_{1,k}+c_{2,k}S_k+c_{3,k}S_k^2\right\}
\end{equation}
where the stochastic coefficients $\left(c_{1,k}\right)_{k\in \Tc}$, $\left(c_{2,k}\right)_{k\in \Tc}$ and $\left(c_{3,k}\right)_{k\in \Tc}$ satisfy the system of coupled stochastic backward recursions given for all $n\in\{1,\dots,T-1\}$ by
\begin{eqnarray*}
c_{1,n-1}&:=&-a_{0,n-1}+c_{1,n}+c_{2,n}\kappa_{n-1}+c_{3,n}\kappa_{n-1}^2+\log\left(\left(1-2c_{3,n}\sigma^2_{n-1}\right)^{-1/2}\right)\\
&&+\frac{c_{2,n}^2\sigma_{n-1}^2+4\kappa_{n-1}^2\sigma_{n-1}^2}{2\left(1-2c_{3,n}\sigma^2_{n-1}\right)}+\frac{2c_{2,n}\sigma_{n-1}^2\kappa_{n-1}}{\left(1-2c_{3,n}\sigma^2_{n-1}\right)},\nonumber\\
c_{2,n-1}&:=&-a_{1,n-1}+c_{2,n}\mu_{n-1}+c_{3,n}\kappa_{n-1}\mu_{n-1}+\frac{4\kappa_{n-1}\mu_{n-1}\sigma_{n-1}^2+2c_{2,n}\sigma_{n-1}^2\mu_{n-1}}{\left(1-2c_{3,n}\sigma^2_{n-1}\right)},\\
c_{3,n-1}&:=&-a_{2,n-1}+c_{3,n}\mu_{n-1}^2+\frac{2\mu_{n-1}^2\sigma_{n-1}^2}{\left(1-2c_{3,n}\sigma^2_{n-1}\right)}.
\end{eqnarray*}
with terminal conditions $c_{1,T}=c_{2,T}=c_{3,T}=0$
\end{Theorem}
\begin{proof}
We will prove this result by backward induction. Thus since $\tilde{P}(T,T)=1$, the exponential quadratic term structure \eqref{EqPropoFullHis} is true for $k=T$. Assume now that the result holds for $k=n$, we would like to prove that this result also holds for $k=n-1$. Hence, by the Definition \eqref{EqP2} and iterated conditional expectation, we obtain 
\begin{eqnarray*}
\tilde{P}(n-1,T)&=&\E\left[\exp\left(-\sum_{t=n-1}^{T-1}r_t\right)\vert \tilde{\Gc}_{n-1}\right]=\E\left[\E\left[\exp\left(-\sum_{t=n-1}^{T-1}r_t\right)\vert \tilde{\Gc}_{n}\right]\vert\tilde{\Gc}_{n-1}\right],\\
&=&\E\left[\exp\left(-r_{n-1}\right)\E\left[\exp\left(-\sum_{t=n}^{T-1}r_t\right)\vert \tilde{\Gc}_{n}\right]\vert\tilde{\Gc}_{n-1}\right],\\
&=&\exp\left(-r_{n-1}\right)\E\left[\E\left[\exp\left(-\sum_{t=n}^{T-1}r_t\right)\vert \tilde{\Gc}_{n}\right]\vert\tilde{\Gc}_{n-1}\right].\\
\end{eqnarray*}
we can use the assumption that the exponential quadratic term structure \eqref{EqPropoFullHis} holds for $k=n$. Thus, we get using also \eqref{EqR} and \eqref{EqSDis},
\begin{eqnarray*}
\tilde{P}(n-1,T)\hspace{-0.3cm}&=&\hspace{-0.3cm}\exp\left(-r_{n-1}\right)\E\left[\exp\left(c_{1,n}+c_{2,n}S_n+S_nc_{3,n}S_n^2\right)\vert\tilde{\Gc}_{n-1}\right],\\
&=&\hspace{-0.3cm}\exp\left\{-a_{0,n-1}-a_{1,n-1}S_{n-1}-a_{2,n-1}S_{n-1}^2\right\}\times\\
&&\hspace{-0.3cm}\E\left[\exp\left\{c_{1,n}+c_{2,n}\left(\kappa_{n-1}+\mu_{n-1}S_{n-1}+\sigma_{n-1}\varepsilon_{n}\right)+c_{3,n}\left(\kappa_{n-1}+\mu_{n-1}S_{n-1}+\sigma_{n-1}\varepsilon_{n}\right)^2\right\}\vert\tilde{\Gc}_{n-1}\right],\\
&=&\hspace{-0.3cm}\exp\left\{-a_{0,n-1}-a_{1,n-1}S_{n-1}-a_{2,n-1}S_{n-1}^2\right\}\times\\
&&\hspace{-0.3cm}\E\left[\exp\left\{c_{1,n}+c_{2,n}\left(\kappa_{n-1}+\mu_{n-1}S_{n-1}\right)+c_{2,n}\sigma_{n-1}\varepsilon_{n}
+c_{3,n}\left(\kappa_{n-1}+\mu_{n-1}S_{n-1}\right)^2\right.\right.\\
&&\left.\left.+c_{3,n}\sigma^2_{n-1}\varepsilon_{n}^2+2\left(\kappa_{n-1}+\mu_{n-1}S_{n-1}\right)\sigma_{n-1}\varepsilon_{n}\right\}\vert\tilde{\Gc}_{n-1}\right],
\end{eqnarray*}
\begin{eqnarray*}
&=&\hspace{-0.3cm}\exp\left\{-a_{0,n-1}-a_{1,n-1}S_{n-1}-a_{2,n-1}S_{n-1}^2+c_{1,n}+c_{2,n}\left(\kappa_{n-1}+\mu_{n-1}S_{n-1}\right)+c_{3,n}\left(\kappa_{n-1}+\mu_{n-1}S_{n-1}\right)^2\right\}\times\\
&&\hspace{-0.3cm}\E\left[\exp\left\{c_{2,n}\sigma_{n-1}\varepsilon_{n}+c_{3,n}\sigma^2_{n-1}\varepsilon_{n}^2+2\left(\kappa_{n-1}+\mu_{n-1}S_{n-1}\right)\sigma_{n-1}\varepsilon_{n}\right\}\vert\tilde{\Gc}_{n-1}\right],\\
&=&\hspace{-0.3cm}\exp\left\{-a_{0,n-1}-a_{1,n-1}S_{n-1}-a_{2,n-1}S_{n-1}^2+c_{1,n}+c_{2,n}\left(\kappa_{n-1}+\mu_{n-1}S_{n-1}\right)+c_{3,n}\left(\kappa_{n-1}+\mu_{n-1}S_{n-1}\right)^2\right\}\times\\
&&\hspace{-0.3cm}\E\left[\exp\left\{f_n\varepsilon_{n}+g_n\varepsilon_{n}^2\right\}\vert\tilde{\Gc}_{n-1}\right],
\end{eqnarray*}
%
with
\begin{eqnarray*}
f_n&:=&c_{2,n}\sigma_{n-1}+2\left(\kappa_{n-1}+\mu_{n-1}S_{n-1}\right)\sigma_{n-1},\\
g_n&:=&c_{3,n}\sigma^2_{n-1}.
\end{eqnarray*}
Since $\varepsilon:=(\varepsilon_k)_{k\in \{1,2,\dots,T\}}$ are a sequence of independent and identically distributed random variables with law $\Nc(0,1)$, we have that
\begin{eqnarray}\label{EqTEmp1}
\E\left[\exp\left(f_n\varepsilon_{n}+g_n\varepsilon_{n}^2\right)\vert\tilde{\Gc}_{n-1}\right]&=&\int_{\R}\exp\left(f_n\varepsilon_{n}+g_n\varepsilon_{n}^2\right)\frac{1}{(2\pi)^{1/2}}\exp\left(-\frac{1}{2}\varepsilon_{n}^2\right)d\varepsilon_{n},\nonumber\\
&=&\frac{1}{(2\pi)^{1/2}}\int_{\R}\exp\left(f_n\varepsilon_{n}+g_n\varepsilon_{n}^2-\frac{1}{2}\varepsilon_{n}^2\right)d\varepsilon_{n}.
\end{eqnarray}
Moreover, we have that
\begin{eqnarray*}
f_n\varepsilon_{n}+g_n\varepsilon_{n}^2-\frac{1}{2}\varepsilon_{n}^2&=&-\frac{1}{2}\left[\left(\left(1-2g_n\right)^{1/2}\varepsilon_{n}-\left(1-2g_n\right)^{-1/2}f_n\right)^2-\left(1-2g_n\right)^{-1}f_n^2\right].
\end{eqnarray*}
Thus, denoting by $\delta_n:=\left(1-2g_n\right)^{-1/2}$, we obtain
\begin{eqnarray*}
f_n\varepsilon_{n}+g_n\varepsilon_{n}^2-\frac{1}{2}\varepsilon_{n}^2&=&-\frac{1}{2}\left[\left(\delta_n^{-1}\varepsilon_{n}-\delta_nf_n\right)^2-\delta_n^{2}f_n^2\right].
\end{eqnarray*}
Replacing this formula into \eqref{EqTEmp1} gives
\begin{eqnarray*}
\E\left[\exp\left(f_n\varepsilon_{n}+g_n\varepsilon_{n}^2\right)\vert\tilde{\Gc}_{n-1}\right]&=&\frac{1}{(2\pi)^{1/2}}\int_{\R}\exp\left(-\frac{1}{2}\left[\left(\delta_n^{-1}\varepsilon_{n}-\delta_nf_n\right)^2-\delta_n^{2}f_n^2\right]\right)d\varepsilon_{n}\\
&=&\frac{1}{(2\pi)^{1/2}}\int_{\R}\exp\left(-\frac{1}{2}\left(\delta_n^{-1}\varepsilon_{n}-\delta_nf_n\right)^2+\frac{1}{2}\delta_n^{2}f_n^2\right)d\varepsilon_{n}\\
&=&\frac{1}{(2\pi)^{1/2}}\exp\left(\frac{1}{2}\delta_n^{2}f_n^2\right)\int_{\R}\exp\left(-\frac{\left(\delta_n^{-1}\varepsilon_{n}-\delta_nf_n\right)^2}{2}\right)d\varepsilon_{n}\\
&=&\delta_n\exp\left(\frac{\delta_n^{2}f_n^2}{2}\right)\frac{1}{\delta_n(2\pi)^{1/2}}\int_{\R}\exp\left(-\frac{1}{2}\left(\frac{\varepsilon_{n}-\delta_n^2f_n}{\delta_n}\right)^2\right)d\varepsilon_{n}\\
&=&\delta_n\exp\left(\frac{\delta_n^{2}f_n^2}{2}\right).
\end{eqnarray*}
Hence we obtain that
\begin{eqnarray*}
\tilde{P}(n-1,T)\hspace{-0.3cm}&=&\hspace{-0.3cm}\exp\left\{-a_{0,n-1}-a_{1,n-1}S_{n-1}-a_{2,n-1}S_{n-1}^2+c_{1,n}+c_{2,n}\left(\kappa_{n-1}+\mu_{n-1}S_{n-1}\right)\right\}\\
&&\hspace{-0.3cm}\times\exp\left\{c_{3,n}\left(\kappa_{n-1}+\mu_{n-1}S_{n-1}\right)^2\right\}\delta_n\exp\left\{\frac{\delta_n^{2}f_n^2}{2}\right\},\\
&=&\hspace{-0.3cm}\exp\left\{-a_{0,n-1}-a_{1,n-1}S_{n-1}-a_{2,n-1}S_{n-1}^2+c_{1,n}+c_{2,n}\left(\kappa_{n-1}+\mu_{n-1}S_{n-1}\right)\right\}\\
&&\hspace{-0.3cm}\times\exp\left\{c_{3,n}\left(\kappa_{n-1}+\mu_{n-1}S_{n-1}\right)^2\right\}\left(1-2g_n\right)^{-1/2}\exp\left\{\frac{\left(1-2g_n\right)^{-1}f_n^2}{2}\right\},\\
&=&\hspace{-0.3cm}\exp\left\{-a_{0,n-1}-a_{1,n-1}S_{n-1}-a_{2,n-1}S_{n-1}^2+c_{1,n}+c_{2,n}\left(\kappa_{n-1}+\mu_{n-1}S_{n-1}\right)\right\}\\
&&\hspace{-0.3cm}\times\exp\left\{c_{3,n}\left(\kappa_{n-1}+\mu_{n-1}S_{n-1}\right)^2\right\}\left(1-2c_{3,n}\sigma^2_{n-1}\right)^{-1/2}\\
&&\hspace{-0.3cm}\times\exp\left\{\frac{\left(1-2c_{3,n}\sigma^2_{n-1}\right)^{-1}\left(c_{2,n}\sigma_{n-1}+2\left(\kappa_{n-1}+\mu_{n-1}S_{n-1}\right)\sigma_{n-1}\right)^2}{2}\right\},\\
&=&\hspace{-0.3cm}\exp\left\{-a_{0,n-1}-a_{1,n-1}S_{n-1}-a_{2,n-1}S_{n-1}^2+c_{1,n}+c_{2,n}\kappa_{n-1}+c_{2,n}\mu_{n-1}S_{n-1}\right\}\\
&&\hspace{-0.3cm}\times\exp\left\{c_{3,n}\kappa_{n-1}^2+c_{3,n}\mu_{n-1}^2S_{n-1}^2+c_{3,n}\kappa_{n-1}\mu_{n-1}S_{n-1}\right\}\\
&&\hspace{-0.3cm}\times\exp\left\{\log\left(\left(1-2c_{3,n}\sigma^2_{n-1}\right)^{-1/2}\right)\right\}
\exp\left\{\frac{c_{2,n}^2\sigma_{n-1}^2+4\kappa_{n-1}^2\sigma_{n-1}^2+4\mu_{n-1}^2S_{n-1}^2\sigma_{n-1}^2}{2\left(1-2c_{3,n}\sigma^2_{n-1}\right)}\right\}\\
&&\hspace{-0.3cm}\times\exp\left\{\frac{8\kappa_{n-1}\mu_{n-1}S_{n-1}\sigma_{n-1}^2+4c_{2,n}\sigma_{n-1}^2\kappa_{n-1}+4c_{2,n}\sigma_{n-1}^2\mu_{n-1}S_{n-1}}{2\left(1-2c_{3,n}\sigma^2_{n-1}\right)}\right\},\\
&=&\hspace{-0.3cm}\exp\left\{-a_{0,n-1}+c_{1,n}+c_{2,n}\kappa_{n-1}+c_{3,n}\kappa_{n-1}^2+\log\left(\left(1-2c_{3,n}\sigma^2_{n-1}\right)^{-1/2}\right)\right.\\
&&\hspace{-0.3cm}\left.+\frac{c_{2,n}^2\sigma_{n-1}^2+4\kappa_{n-1}^2\sigma_{n-1}^2}{2\left(1-2c_{3,n}\sigma^2_{n-1}\right)}+\frac{4c_{2,n}\sigma_{n-1}^2\kappa_{n-1}}{2\left(1-2c_{3,n}\sigma^2_{n-1}\right)}\right\}\\
&&\hspace{-0.3cm}\exp\left\{S_{n-1}\left(-a_{1,n-1}+c_{2,n}\mu_{n-1}+c_{3,n}\kappa_{n-1}\mu_{n-1}+\frac{8\kappa_{n-1}\mu_{n-1}\sigma_{n-1}^2+4c_{2,n}\sigma_{n-1}^2\mu_{n-1}}{2\left(1-2c_{3,n}\sigma^2_{n-1}\right)}\right)\right\}\\
&&\hspace{-0.3cm}\exp\left\{S_{n-1}^2\left(-a_{2,n-1}+c_{3,n}\mu_{n-1}^2+\frac{4\mu_{n-1}^2\sigma_{n-1}^2}{2\left(1-2c_{3,n}\sigma^2_{n-1}\right)}\right)\right\}.
\end{eqnarray*}
Thus, by identification, we get that
\begin{eqnarray*}
c_{1,n-1}&:=&-a_{0,n-1}+c_{1,n}+c_{2,n}\kappa_{n-1}+c_{3,n}\kappa_{n-1}^2+\log\left(\left(1-2c_{3,n}\sigma^2_{n-1}\right)^{-1/2}\right)\\
&&+\frac{c_{2,n}^2\sigma_{n-1}^2+4\kappa_{n-1}^2\sigma_{n-1}^2}{2\left(1-2c_{3,n}\sigma^2_{n-1}\right)}+\frac{4c_{2,n}\sigma_{n-1}^2\kappa_{n-1}}{2\left(1-2c_{3,n}\sigma^2_{n-1}\right)},\\
c_{2,n-1}&:=&-a_{1,n-1}+c_{2,n}\mu_{n-1}+c_{3,n}\kappa_{n-1}\mu_{n-1}+\frac{8\kappa_{n-1}\mu_{n-1}\sigma_{n-1}^2+4c_{2,n}\sigma_{n-1}^2\mu_{n-1}}{2\left(1-2c_{3,n}\sigma^2_{n-1}\right)},\\
c_{3,n-1}&:=&-a_{2,n-1}+c_{3,n}\mu_{n-1}^2+\frac{4\mu_{n-1}^2\sigma_{n-1}^2}{2\left(1-2c_{3,n}\sigma^2_{n-1}\right)}.
\end{eqnarray*}
and so the expected result.
\end{proof}

Note that regarding \eqref{EqP2}, $\tilde{P}(k,T)$ is a function of the history of the Markov chain $X$ between time $k$ and $T-1$. Thus we can write $\tilde{P}(k,T,X_k,X_{k+1},\dots,X_{T-1})$. Moreover, the coefficients $c_{1,k},c_{2,k}$ and $c_{3,k}$, $k\in \{0,1,\dots,T-1\}$ are measurable with respect to the $\sigma$-algebra generated by $X_k,X_{k+1},\dots,$ and $X_{T-1}$. So they can be represented as functions of them. Hence we obtain for $k \in \{0,1,2,\dots,T-1\}$
\begin{eqnarray*}
c_{1,k}:=c_1(k,X_k)=c_1(k,X_k,X_{k+1},\dots,X_{T-1}),\\
c_{2,k}:=c_2(k,X_k)=c_2(k,X_k,X_{k+1},\dots,X_{T-1}),\\
c_{3,k}:=c_3(k,X_k)=c_3(k,X_k,X_{k+1},\dots,X_{T-1}).
\end{eqnarray*}
This means by given $\tilde{\Gc}_k:=\Fc_T^X\vee \Fc^S_k$, the conditional bond price $\tilde{P}(k,T,X_k,X_{k+1},\dots,X_{T-1})$ can be represented as follows:
\begin{eqnarray}\label{EqP3}
&&\tilde{P}(k,T,X_k,X_{k+1},\dots,X_{T-1})=\\
&&\exp\left\{c_1(k,X_k,X_{k+1},\dots,X_{T-1})+c_2(k,X_k,X_{k+1},\dots,X_{T-1})S_k+c_3(k,X_k,X_{k+1},\dots,X_{T-1})S_k^2\right\}.\nonumber
\end{eqnarray}

\begin{Remark}
In the specific case of an affine term structure of interest rate (i.e. $a_{2,k}\equiv 0$ in \eqref{EqR}), we have that 
\begin{equation}\label{EqRAffine}
r_k:=r(k,X_k)=a_{0,k}+a_{1,k}S_k, \quad k\in \Tc.
\end{equation}
And so we get that the conditional bond price $\tilde{P}(k,T)$ admits also a affine structure form 
\begin{equation}\label{EqPropoFullHis}
\tilde{P}(k,T)=\exp\left\{c_{1,k}+c_{2,k}S_k\right\}, \quad k\in \Tc,
\end{equation}
where coefficient $c_{1,k}$ and $c_{2,k}$ satisfy the system of coupled stochastic backward recursions given for all $n\in\{1,\dots,T-1\}$ by
\begin{eqnarray*}
c_{1,n-1}&:=&-a_{0,n-1}+c_{1,n}+c_{2,n}\kappa_{n-1}+\frac{c_{2,n}^2\sigma_{n-1}^2}{2}+2\kappa_{n-1}^2\sigma_{n-1}^2+2c_{2,n}\sigma_{n-1}^2\kappa_{n-1},\nonumber\\
c_{2,n-1}&:=&-a_{1,n-1}+c_{2,n}\mu_{n-1}+4\kappa_{n-1}\mu_{n-1}\sigma_{n-1}^2+2c_{2,n}\sigma_{n-1}^2\mu_{n-1}.
\end{eqnarray*}
with terminal conditions $c_{1,T}=c_{2,T}=0$ (see Duffie and Kan (1996) in \cite{DK96} for more details about affine interest rate structure).

\end{Remark}

\subsection{General case}

In practice, we do not know the full history of the Markov chain $X$. Indeed, we do not know all the future states of the economy. So we need to evaluate our bond price given only the information set $\Gc_k$. Hence, following the representation \eqref{EqP3} and the Theorem \ref{PropoFullHis} we obtain the following result:

\begin{Proposition}\label{PropoG}
Under the information set $\Gc_k$, we have that the Bond price $P$ at time $k\in \Tc$ is given by
\begin{eqnarray}
P(k,T)&=&\sum_{i_k,i_{k+1},\dots,i_{T-1}=1}^N\left[\prod_{l=k}^{T-1}q_{i_{l}i_{l+1}k}\right]\tilde{P}(k,T,e_{i_{k}},e_{i_{k+1}},\dots,e_{i_{T-1}})
\end{eqnarray}
where $\tilde{P}$ is given by \eqref{EqP3} and coefficients $c_i(k,X_k,X_{k+1},\dots,X_{T-1})$ for $i=\{1,2,3\}$ follow the recursive system given in Theorem \ref{PropoFullHis}. 
\end{Proposition}
\begin{proof}
This result is obtained from taking the expectation of $\tilde{P}(k,T)$ conditioning on $\Gc_k$ and by enumerating all the transitions probabilities of the Markov chain $X$ from time $k$ to $T-1$.
\end{proof}

%
%
%
%

\section{Conclusion}
We prove that if the short term interest rate follows a quadratic term structure of a regime switching asset process then the conditional zero coupon bond price with respect to the Markov switching process admits also a quadratic term structure. Moreover, the stochastic coefficients appearing in this quadratic decomposition satisfy an explicit system of coupled stochastic backward recursions. This allows us to obtain an explicit way to evaluate this conditional zero coupon bond price.

\end{document}